\documentclass[sigconf]{aamas}  

\usepackage{booktabs}
\usepackage[linesnumbered,ruled,vlined]{algorithm2e}
\usepackage{color}
\usepackage{colortbl}
\usepackage{amsmath,amsthm}
\newtheorem{prop}{Proposition}
\newtheorem{mydef}{Definition}
\usepackage{adjustbox}
\newcommand{\compas}{{\tt ComPAS}}
\newcommand\TODO[1]{\textcolor{red}{#1}}

\usepackage{multirow}

\newenvironment{function1}[1][htb]
  {
   \begin{algorithm}
  }{\end{algorithm}}

\newenvironment{function2}[1][htb]
  {
   \begin{algorithm}
  }{\end{algorithm}}  
  
  \newenvironment{function3}[1][htb]
  {
   \begin{algorithm}
  }{\end{algorithm}}

\newenvironment{function4}[1][htb]
  {
   \begin{algorithm}
  }{\end{algorithm}}  

\newenvironment{function5}[1][htb]
  {
   \begin{algorithm}
  }{\end{algorithm}}
  
  \newenvironment{function6}[1][htb]
  {
   \begin{algorithm}
  }{\end{algorithm}}
  
  \newenvironment{function7}[1][htb]
  {
   \begin{algorithm}
  }{\end{algorithm}}
  \newenvironment{function8}[1][htb]
  {
   \begin{algorithm}
  }{\end{algorithm}}
\newcommand{\noteng}[1]{\textcolor{red}{NG: #1}}

\setcopyright{ifaamas}  
\acmDOI{doi}  
\acmISBN{}  
\acmConference[AAMAS'18]{Proc.\@ of the 17th International Conference on Autonomous Agents and Multiagent Systems (AAMAS 2018), M.~Dastani, G.~Sukthankar, E.~Andre, S.~Koenig (eds.)}{July 2018}{Stockholm, Sweden}  
\acmYear{2018}  
\copyrightyear{2018}  
\acmPrice{}  

\begin{document}

\title{ComPAS: Community Preserving Sampling\\ for Streaming Graphs}  




%
\author{Sandipan Sikdar}
\affiliation{%
  \institution{IIT Kharagpur}
}
\email{sandipansikdar@cse.iitkgp.ernet.in}
\author{Tanmoy Chakraborty}
\affiliation{%
  \institution{IIIT Delhi}
}
\email{tanmoy@iiitd.ac.in}
\author{Soumya Sarkar}
\affiliation{%
  \institution{IIT Kharagpur}
}
\email{soumya@iitkgp.ac.in}
\author{Niloy Ganguly}
\affiliation{%
  \institution{IIT Kharagpur}
}
\email{niloy@cse.iitkgp.ernet.in}

\author{Animesh Mukherjee} 
\affiliation{%
 \institution{IIT Kharagpur}
}
\email{animeshm@cse.iitkgp.ernet.in}
%
%
%
%

\begin{abstract}  
In the era of big data, graph sampling is indispensable in many settings. Existing sampling methods are mostly designed for static graphs, and aim to preserve {\em basic} structural properties of the original graph  (such as degree distribution, clustering coefficient etc.) in the sample. We argue that for any sampling method it is impossible to produce an universal representative sample which can preserve {\em all} the properties of the original graph; rather sampling should be application specific (such as preserving hubs - needed for information diffusion). Here we consider {\em community detection} as an application scenario. We propose \compas, a novel sampling strategy that unlike previous methods, is not only designed for {\em streaming graphs} (which is a more realistic 
representation of a real-world scenario) 
but also {\em preserves the community structure} of the original graph in the sample. 
Empirical results on both synthetic and different real-world graphs show that \compas~is the best to preserve the underlying community structure with average performance reaching 73.2\% of the most informed algorithm for {\em static} graphs. 

\end{abstract}

%

\keywords{Streaming graph; Sampling; Community detection}  

\maketitle


\section{Introduction}

One of the fundamental techniques to analyze very large-scale graphs is through sampling~\cite{leskovec2006sampling}, especially where the analysis on the entire graph is intractable (and often impractical).
A good sampling method should usually target a specific application and essentially preserve a set of (not all) properties of the original graph geared toward the application. For instance, a sampling method
designed for information diffusion should preserve the hubs (high-degree nodes) in the sample; whereas, a sampling scheme for outbreak detection (such as disease outbreak) should preserve the nodes with high local clustering coefficient.
Sampling has been studied extensively in the context of {\em static graphs}~\cite{leskovec2006sampling,gjoka2010walking,maiya2010sampling,rasti2009respondent,ribeiro2010estimating}; 
however, there has been very limited work on sampling from {\em streaming graphs}~\cite{henzinger1998computing} where nodes/edges arrive in discrete time intervals and only a part of the entire graph is available for analysis at any point of time~\cite{aggarwal2011outlier,ahmed2014network,lim2015mascot,de2016tri}. 

Existing graph sampling methods are mostly designed for preserving simple structural properties 
(such as degree distribution, clustering coefficient etc.) of the original graph in the sample - only few works attempted
 to
preserve complex properties like community~\cite{tong2016novel,maiya2010sampling} - which may be useful for designing a wide range of applications.
For instance, in marketing, surveys often seek to construct 
samples from different communities to capture  the diversity of the population (also known as \textit{cluster sampling})~\cite{kolascyk2013statistical}. 
In this paper, we propose a novel sampling algorithm that preserves the original {\em community structure}\footnote{In this paper, we consider disjoint community structure.} of streaming graphs. Our work sharply contrasts the recently proposed Green Algorithm (GA)~\cite{tong2016novel} which, is explicitly designed to generate a sample that preserves the community structure for {\em static graphs}. 

\noindent{\bf Our contributions:} In this paper, we propose~\compas, a novel sampling algorithm on streaming graph (most realistic graph representation~\cite{aggarwal2011outlier,ahmed2014network}) that is capable of
 preserving 
 the community structure of the original graph. 
 
\compas~is designed based on a novel hypothesis that  graph sampling and community detection can be \textit{interwoven} together 
 to produce a more representative sample. In particular, our contributions in this paper are the following: 
\begin{itemize}
\item To the best of our knowledge \compas~is the first {\em community-preserving} sampling method for {\em streaming graphs}. Along with the sample nodes, \compas~also outputs the community structure of the sample that closely corresponds to the community structure of the original graph.

\item In absence of any other community preserving sampling algorithm for streaming graphs, we resort to comparing \compas~ with GA~\cite{tong2016novel} which was designed to preserve the community structure while sampling from \textit{static} graphs. Note that GA, unlike \compas, has the information of the full graph while sampling and building the community structure. Empirical evidences on synthetic and real-world graphs demonstrate that the sample generated by \compas~correctly preserves the community structure 
with average performance reaching as high as 73.2\% of GA. Further, we also compare \compas~with well-known node/edge preserving sampling methods available for streaming graphs to show that these do not automatically preserve the community structure thus necessitating the design of \compas~. 

\item We  do a detailed micro-analysis to comprehend the reasons behind superior performance of \compas.
We also show additional benefits of \compas~through an application -- 
 selection of (limited) training set for online learning. We obtain a performance that is within 90.5\% of the most informed algorithm GA available for static graphs. 

\end{itemize}


\section{Related work}
Population sampling has been studied for long in social sciences 
\cite{frank1977survey},\cite{frank1980sampling}, 
such as snowball sampling  \cite{goodman1961snowball},  respondent-driven sampling  \cite{heckathorn1997respondent}, \cite{gile2010respondent} etc. 
and most of the relevant works in this space deal with estimating global properties 
of the population (see a survey in~\cite{kolascyk2013statistical}).

\noindent{\bf Sampling from static graphs:} Availability of large-scale graph data has generated renewed interest in the sampling problem ~\cite{leskovec2006sampling,gjoka2010walking,rasti2009respondent,ribeiro2010estimating,ahmed2010reconsidering}. Following in this series are works like~\cite{maiya2010sampling} and~\cite{maiya2011benefits}. A severe limitation of these approaches is that they assume that the entire graph is present in advance (i.e., the snapshot is static) for the algorithm to produce the desired output.

 
\noindent{\bf Sampling from streaming graphs:} With increasing interest in mining and analysis of large social graphs (which are mostly dynamic in nature), 
there is a recent shift in focus toward sampling from streaming graphs.
A streaming graph corresponds to a \textit{stream of incoming edges} (see Figure  \ref{fig_algo}).~\cite{aggarwal2011outlier} proposed a 
streaming edge sampling (SE) algorithm for outlier detection.~\cite{ahmed2014network} proposed streaming node sampling (SN), 
streaming BFS (Breadth First Search, SBFS) and Partially Induced Edge Sampling (PIES) algorithms. SN and SE maintain a reservoir of nodes and edges respectively and 
insert or remove them based on a pre-defined hash function. While SBFS essentially implements simple breadth-first search on a sliding window of fixed number of 
edges in the stream, PIES leverages a partial induction of nodes and combines edge-based node sampling with the graph
induction in a single pass.
Other recent works include~\cite{lim2015mascot,de2016tri}.

\noindent{\bf The most informative baseline:} The Green Algorithm (GA)~\cite{tong2016novel} is capable of generating community structure preserving samples 
for {\bf static graph}; however, the explicit community structure is not produced as an output of the algorithm. 
This constitutes, for us, the \textit{most informative} baseline since it has to have the full original network at its disposal to decide whether to include an edge or node in the sample it constructs.
Typically the set of nodes with high clustering coefficients as well as high degree are sampled in. 
However, for a streaming graph setting, this exercise becomes  difficult as one needs to determine the importance of a node based only on its  limited arrival history. 
To this aim we incorporate a simple technique which allows \compas~to correctly identify the high fidelity (high degree and clustering coefficient)  nodes and, thereby, improve the quality of the sample. 
Moreover, we intend to create samples in such a way that the nodes thus sampled are largely connected among themselves.
{This may be specially important for problems where edge characteristics are necessary like link prediction~\cite{de2013discriminative}, epidemic flow modeling~\cite{li2013influence}, signed network friend/foe classification~\cite{leskovec2010predicting}}. 
\if{0}
In contrast, \compas~ 
decides the inclusion of nodes/edges on the fly as they arrive, i.e., \textit{without any knowledge} of the full original network for sampling fulfilling the constrain of streaming graphs. 
{Typically the nodes with high clustering coefficients as well as high degree are chosen as community centers and this could be easily identified in case of a static network. 
For streaming graph setting, this exercise becomes increasingly difficult as one needs to determine the importance of a node based only on its arrival history which too is limited (not all arrivals can be recorded). To this aim we incorporate a simple technique which allows \compas~to correctly identify the high fidelity nodes and thereby improving the quality of the sample. Our proposed framework is also intended at capturing the giant components in the graph.}
\fi

\section{Problem definition}
\label{prob_def}

We consider a graph stream $S$ represented by a set of edges $e_1$, $e_2$, $\cdots$ with each edge $e_i$ arriving at 
$i^{th}$ (discrete) time step. A graph $G$ at time $t$ is the aggregate of all the edges arriving till time $t$. 
$V$ represents the set of unique nodes present in  $G$. The community structure of $G$ is represented by $C$. We consider $G$ to be both unweighted and undirected. 


\begin{mydef}
Given a streaming graph $G$ of size $V$, our objective is to obtain a sample graph $G_s$ of size $n$ 
 such that $C$, the underlying community structure of $G$ is highly preserved in $G_s$  $($i.e., $C \sim C_s$ where $C_s$ is the community structure of $G_s)$ given the constrain that any
algorithm at any discrete time step can only utilize the information of 
last arrived $\mathcal{H}$ ($<<$ $V$, $n$) nodes (which is maintained in a buffer). 
\end{mydef}

\setlength{\textfloatsep}{-3pt}

\begin{algorithm}[!h]
\tiny
\caption{\compas: A {\bf Com}munity {\bf P}reserving Sampling {\bf A}lgorithm for {\bf S}treaming Graph}\label{alg:compas}
\KwData{$S$: Graph stream, $n$: Sample size,  $\alpha$: Initial fraction of nodes inserted, $n_d$: size of the buffer, $Algo$: a community detection algorithm}\label{algo_compas}
\KwResult{Sampled subgraph $G_s(V_s, E_s)$, $C_s$}
Initialize $G_s$: $V_s=\phi$, $E_s=\phi$\\
Create an empty buffer $\mathcal{H}$ of size $n_d$\\
Initialize buffer $\mathcal{H}$: $\mathcal{H}_c=\phi$, $\mathcal{H}_p=\phi$\\
$flag=1$, $t=0$\\
\For{$e_t$ in the graph stream $S$}{
$e_t=\{u,v\}$\\
\If{$\frac{|V_s|}{n}<\alpha \wedge e_t\notin E_s$}{\label{b:step1}
$V_s=V_s\cup u \cup v$\\
$E_s=E_s\cup e_t$\\
\textbf{Continue;}\label{e:step1}
}
\ElseIf{flag==1}{
Run $Algo$ on $G_s$ and detect community structure $C_s$\label{algo} \\
$flag$=0\\
}
\ElseIf{$u, v\in V_s$}{ 
$V_s,E_s,C_s=BothinSample(u,v,e_t,V_s,E_s,C_s)$\\
}
\ElseIf{$u,v\notin V_s \wedge u,v\in \mathcal{H}$}{
$\mathcal{H} = NodeinBuffer(u,\mathcal{H})$\\
$\mathcal{H} = NodeinBuffer(v,\mathcal{H})$
}
\ElseIf{$u \in V_s \wedge v\notin V_s \wedge v\in \mathcal{H}$}{
$\mathcal{H} = NodeinBuffer(v,\mathcal{H})$\\
}
\ElseIf{$u\in V_s \wedge v\notin V_s \wedge v\notin \mathcal{H}$}{
$V_s,E_s,C_s,\mathcal{H}=NodeisNew(v,u,V_s,E_s,\mathcal{H},C_s)$
}
\ElseIf{$u\notin V_s \wedge u\in \mathcal{H} \wedge v\notin V_s \wedge v\notin \mathcal{H}$}{
$\mathcal{H}_c(u) = \mathcal{H}_c(u) +1$\\
$V_s,E_s,C_s,\mathcal{H}=NodeisNew(v,u,V_s,E_s,\mathcal{H},C_s)$
}
\ElseIf{$u,v\notin V_s \wedge u,v\notin \mathcal{H}$}{
$V_s,E_s,C_s,\mathcal{H}=NodeisNew(u,v,V_s,E_s,\mathcal{H},C_s)$\\
$V_s,E_s,C_s,\mathcal{H}=NodeisNew(v,u,V_s,E_s,\mathcal{H},C_s)$
}
$t=t+1$
}

\Return $G_s$, $C_s$
\end{algorithm}


\section{Proposed algorithm: \compas}
\label{algorithm}

We propose \compas, a {\bf Com}munity {\bf P}reserving sampling {\bf A}lgorithm for {\bf S}treaming graphs. \compas~aims at sampling a streaming graph in such a way that its underlying community structure is preserved in the sample (Algorithm~\ref{algo_compas} and Figure~\ref{fig_algo} respectively present a pseudo-code and a toy example). 
The algorithm attempts to identify the high fidelity nodes (nodes with high degree and high clustering coefficient) and suitably determine the communities to which they belong. \\
{\bf Description of the algorithm:} To start with, \compas~keeps adding streaming edges (nodes) into the sample $G_s$ as long as a certain number  of nodes ($\alpha \cdot n$, $\alpha~<~1$) are inserted (lines \ref{b:step1}-\ref{e:step1}). 
This constitutes the warm-up knowledge for the structure. 
Once the threshold is reached, a pre-selected community detection algorithm $Algo$ is run on $G_s$ to obtain the {\bf initial community structure} (line \ref{algo}). 

\noindent{\bf Subsequent dynamics:} Henceforth, once an edge $e_t$ is picked up from the stream, \compas~inserts $e_t$ into a buffer $\mathcal{H}$ (size $n_d$) which consists of the two variables -- $\mathcal{H}_c$ and $\mathcal{H}_p$. $\mathcal{H}_c$ counts 
the number of times a node is encountered till that time\footnote{In streaming graph, an edge might appear multiple times in the stream.}, and $\mathcal{H}_p$ keeps track of the current parent of a node (i.e., the node with which it arrived last). 
This presents a crude estimation of the importance of the node, since a recurrently occurring node is probably more important compared to a node occurring only intermittently. The idea is inspired by the reservoir sampling technique introduced in \cite{ahmed2014network}.
Once the buffer $\mathcal{H}$ is full, the insertion activity triggers some chain reactions which are different at the two specific phases 
(a). when the size of $G_s$ is between $\alpha \cdot n$ and $n$  -  any incoming new node triggers the entry of a node($x$) and corresponding edge ($\mathcal {P}(x)$, $x$)  from buffer to $G_s$ and (b). when $G_s$ has already reached $n$ - at that point a node has to be removed from $G_s$ to insert the incoming node~($x$) from buffer. We eliminate the node with least degree and clustering coefficient thus  ensuring progressive
inclusion of high-fidelity node. 

\noindent{\bf Genesis of the six modules of \compas:} Considering differently, for an incoming  streaming edge $e_t~=~\{u,v\}$, each endpoint ($u,v$) could be (i) a new node, 
(ii) present in the buffer or (iii) present in the partially constructed  sample graph $G_s$. Depending on the current position of $u$ and $v$, one of the six conditions is encountered which are consequently handled by the six submodules - 
(1) both endpoints are present in the sample, (2) both are in buffer, (3) one is in sample while the other is in buffer, (4) one is in sample while the other is new, (5) one is in buffer while the other is new and (6) both are new. 
We elaborate on the submodules next. 

\noindent\textbf{(i) \underline{Both $u$ and $v$ are present in the Graph $G_s$:}} When both $u,v$ are in $V_s$, $BothinSample()$ (see Function \ref{bothsample}) is called from 
line 15 of Algorithm \ref{algo_compas}.  
The aim of this module is to place the edge in such a way that the modularity 
of the evolving sample graph ($G_s$) improves. Vis-a-vis the existing community structure, the edge $e_t$ 
can be (a). an intra-community edge (totally inside a single community) or 
(b). an inter-community edge (connecting two communities $C(u)$ and $C(v)$). 

In case of an intra-community edge (edge $\{b,d\}$ in Figure~\ref{fig_algo}), addition of $e_t$ 
increases modularity of the community according to Proposition~\ref{1}\footnote{Detailed proofs of the proposition can be found in section~\ref{appendix}}. Moreover, we also know from Proposition~\ref{2} that 
splitting of current community on addition of a new intra-community edge does not increase modularity~\cite{PhysRevE.78.046115}. 
Therefore we leave $C_s$ in its current form without any modification.

\begin{prop}\label{1}
Addition of an edge to a community $c\in C$, increases its modularity if $D_c\leq M-1$ $($where $M=|E|$ and $D_c$ is total degree of all the nodes $c$ ).
\end{prop}

\if{0}
\begin{proof}
From Equation \ref{modularity}, we see the contribution of individual community $c\in C$ in modularity as: $Q_c=\frac{m_c}{M} - \frac{D_c^2}{4M^2}$. 

Addition of a new edge within $c$, the $c$'s contribution of modularity becomes:
\[
Q'_c=\frac{m_c+1}{M+1} - \frac{(D_c+2)^2}{4(M+1)^2}
\]

So the increase in modularity is $\Delta Q_c=Q'_c-Q_c$,
\[\small
\begin{split}
\Delta Q_c=&\frac{4M^2-4m_cM^2-4D_cM^2-4m_cM+2D_c^2M+D_c^2}{4(M+1)^2M^2}\\
&\geq\frac{(2M^2-2D_cM-D_c)(2M-D_c)}{4(M+1)^2M^2}\\
&\geq 0
\end{split}
\]
The equality holds if $D_c\leq M-1$. This thus implies $(2M^2-2D_cM-D_c)\geq 0$. This proves the proposition.
\end{proof}
\fi

\begin{prop}\label{2}
Addition of any intra-community edge into a community $c\in C$ would not split into smaller communities.
\end{prop}

\if{0}
\begin{proof}
We will prove this proposition by contradiction.
Assume that once a new intra-community edge is added into $c$, it gets split into $k$ small modules, namely $X_1$, $X_2$, $\cdot$,$X_k$. Let $D_{X_i}$ and $e_{ij}$ be the total degree of nodes inside $X_i$ and number of edges connecting $X_i$ and $X_j$ respectively.

Before adding the edge, we have $Q_c \geq \sum_{i=1}^k Q_{X_i}$ (where $Q_c$ is the total modularity of community $c$), because otherwise all $X_i$s can be split earlier, which is not in this case. This implies that: $\frac{m_c}{M}- \frac{D_c^2}{4M^2} > \sum_{i=1}^k (\frac{m_{X_i}}{M} - \frac{D_{X_i}^2}{4M^2})$. Since $X_1,X_2,\cdot,X_k$ are all disjoint modules of $c$, $D_c=\sum_{i=1}^k D_{X_i}$ and $m_c=\sum_{i=1}^k m_{X_i} + \sum_{i<j} e_{ij}$. This further implies that:
$\sum_{i<j} e_{ij} > \frac{\sum_{i<j} D_{X_i}D_{X_j}}{2M}
$.
 
 Without loss of generality, let us assume that the new edge is added inside $X_1$. 
Since we assume that after adding the new edge into $c$, it gets split into $k$ small modules, the modularity value should increase because of the split. Therefore, 
\[\small
\begin{split}\small
& Q'_c < \sum_{i=1}^k Q_{X_i}\\
& \Leftrightarrow \frac{\sum_{i=1}^k m_{X_i} + \sum_{i<j} e_{ij} + 1}{M+1} - \frac{(\sum_{i=1}^k D_{X_i +2})^2}{4(M+1)^2}\\
& < \frac{\sum_{i=1}^k m_{X_i}+1}{M+1} - \frac{(D_{X_1}+2)^2}{4(M+1)^2} - \sum_{i=2}^k \frac{D_{X_i}^2}{4(M+1)^2}\\
&\Leftrightarrow \sum_{i<i} e_{ij} < \frac{\sum_{i=1}^k D_{X_i} - 2 D_{X_1} + \sum_{i<j} D_{X_i}D_{X_j}}{2(M+1)}
\end{split}
\]
Since $\sum_{i=1}^k D_{X_i} - 2D_{X_1} < 2M$, this implies that 
\[\small
\begin{split}
\frac{\sum_{i<j}D_{X_i}D_{X_j}}{2M}  < \sum_{i<j}e_{ij} 
&< \frac{\sum_{i=1}^k D_{X_i} - 2D_{X_1} + \sum_{i\neq j}{D_{X_i}D_{X_j}}}{2(M+1)}\\
& <\frac{\sum_{i<j}D_{X_i}D_{X_j}}{2M}+1
\end{split}
\]
Therefore, the proposition holds.
\end{proof}

\fi
In case of $e_t$ connecting two different communities (edge $\{b,f\}$ in Figure \ref{fig_algo}), three possibilities may arise - \\
(i) $u$ may leave its current community and join $v$'s community, (ii) $v$ may leave its current community and join $u$'s community and (iii) $u$ and $v$ may leave their current communities and together form a new community. In addition, if the community membership of $u$ (or $v$) is changed, this can also pull out its neighbors to join with it, and some of the neighbors might eventually want to change their memberships as well~\cite{pone.0091431}. To decide we first calculate $\Delta Q(u,C(u),C(v))$ (where $\Delta Q(x,C(x),C(y))$ indicates the change in modularity after assigning $x$ from $C(x)$ to $C(y)$) (case (i)), $\Delta Q(v,C(v),C(u))$ (case (ii)) and $\Delta Q(\{u,v\},\{C(u),C(v)\},C^{\ast})$ ($u$ and $v$ change their current communities to form a new community $C^{\ast}$, case (iii)) and select the case where the change in modularity is maximum. Consequently we let the neighbors (of the node whose community membership is altered by the above action) decide their best move in the similar way. This continues recursively (neighbors of neighbors) until the modularity stabilizes or decreases.  

\begin{figure}[!t]
\centering
\includegraphics[width=\columnwidth]{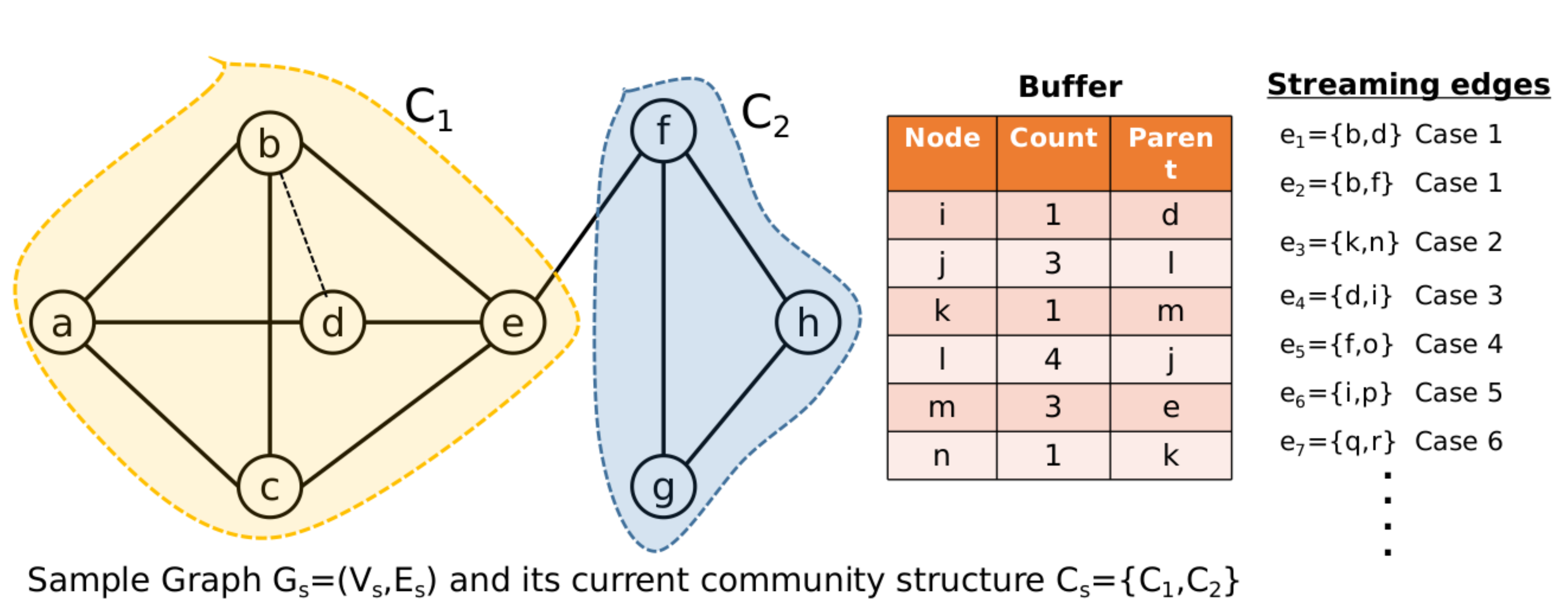}
\vspace{-5mm}
\caption{Toy example depicting various conditions handled by \compas~ when a streaming edge arrives.}\label{fig_algo} 
\vspace{2mm}
\end{figure}


\if{0}
\begin{prop}\label{3}
If a new inter-community edge $(u,v)$ connecting two communities $C(u)$ and $C(v)$ is introduced, $C(u)$ $($or $C(v))$ is the most likely candidate for $v$ (or $u$) if it changes its membership.
\end{prop}

\begin{proof}
The method is inspired by \cite{PhysRevE.78.046115} that a vertex $u$ is influenced by two factors: $F^c_{in}(u)$ the force that keeps $u$ stay in its own community $c$, and $F^c_{out}(u)$, the force that a community $S$ imposes to $u$ in order to bring $u$ to $S$ as follows:
\[\small
F^c_{in}(u)=e_c^u-\frac{d_u(D_c-d_u)}{2M}
\]
and 
\[\small
F^S_{out}(u)=max_{S\in NC(u)} \{e_S^u - \frac{d_u D_{outS}}{2M}\}
\]
where $NC(u)$ is the set of neighboring communities of $u$, and $D_{outS}$ is the total degree of vertices outside $S$.

Now we will show that the presence of new edge $(u,v)$ will strengthen $F^{C(v)}_{out}(u)$ and weaken $F^S_{out}(u)$. In other words, we will show that $F^{C(v)}_{out}(u)$ increases while $F^S_{out}(u)$ decreases for all $S \in C \wedge S \notin \{C(u),C(v)\}$.
\[\small
\begin{split}
&F^{C(v)}_{out}(u)|_{new} - F^{C(v)}_{out}(u)|_{old}\\
&=(e_u^{C(v)}+1-\frac{(d_u+1)(d_{outC(v)}+1)}{2(M+1)} - (e_u^{C(v)} - \frac{d_ud_{outC(v)}}{2M})\\
&\geq \frac{2M+d_ud_{outC(v)}}{2(M+1)} - \frac{d_u d_{outC(v)}+d_{outC(v)}+d_u+1}{2(M+1)}\\
&>0
\end{split}
\]
Therefore $F^{C(v)}_{out}(u)$ is strengthened when a new edge $(u,v)$ is introduced. Further, for any community $S\in C \wedge S\notin \{C(u),C(v)\}$
\[
\begin{split}
& F^{S}_{out}(u)|_{new} - F^{S}_{out}(u)|_{old}\\
&= (e^S_u-\frac{(d_u+1)d_{outS}}{2(M+1)}) - (e^S_u - \frac{d_ud_{outS}}{2M})\\
&= d_{outS}(\frac{d_u}{2M} - \frac{d_u+1}{2(M+1)})<0
\end{split}
\]

This implies that $F^S_{out}(u)$ is weakened when $(u,v)$ is added. Therefore,  the proposition holds.
\end{proof}
\fi

\if{0}
\begin{prop}\label{complex}
If a new edge $(u,v)$ is added into the graph, then joining $u$ to $v$'s  community $C(v)$ will increase the modularity value  if $\Delta Q(u,C(u),C(v)) \equiv 4(M+1)(e_{C(v)}^u+1-e_{C(u)}^u)+e_{C(v)}^u(2D_{C(v)}-2d_u-e_{C(u)}^u) - 2(d_u+1)(d_u+1+d_{C(v)}-d_{C(u)})>0$.
\end{prop}
\fi
\if{0}
\begin{proof}
Vertex $u$ will leave its current community $C(u)$ and join $v$'s community $C(v)$ if
\[\small
\begin{split}
& Q_{C(v)+u} + Q_{C(u)-u} > Q_{c(u)} + Q_{C(v)}\\
&\Leftrightarrow \frac{m_{C(v)}+e_{C(v)}+1}{M+1} - \frac{(d_{C(v)}+d_u+2)^2}{4(M+1)^2}+\\
&\frac{m_{C(u)} - e_{C(u)}}{M+1} - \frac{(d_{C(u)}-d_u-e_{C(u)})}{4(M+1)^2} \\
& > \frac{m_{C(v)}}{M+1} - \frac{(d_{C(v)}+d_u+2)^2}{4(M+1)^2} + \frac{m_{C(u)}}{M+1} - \frac{(d_{C(u)}+1)^2}{4(M+1)^2}\\
&\Leftrightarrow 4(M+1)(e_{C(v)} +1 -e_{C(u)}) + e_{C(u)}(2d_{C(v)} - 2d_{C(u)}-e_{C(u)})\\
& - 2(d_{C(u)}+1)(d_{C(u)}+1+d_{C(v)}-d_{C(u)})>0
\end{split}\]
\end{proof}
\fi
\if{0}
\noindent{\bf Corollary 1.} {\em If the condition in Proposition \ref{complex} is not satisfied, then neither $u$ nor its neighbors should be assigned to $C(v)$.}
\fi

\begin{function1}\tiny
\caption{\small$BothinSample(u,v,e_t,V_s,E_s,C_s)$} 
\label{bothsample} 
\If{$C_s(u)==C_s(v)$}{
$E_s=E_s\cup e_t$
}
\Else{
\If{$\Delta Q (u,C_s(u),C_s(v))<0 \wedge \Delta Q (v,C_s(v),C_s(u))<0 \wedge Q(\{u,v\},\{C(u),C(v)\},C^{\ast})<0$}{
\Return{$V_s,E_s,C_s$}}
\Else{
$w = arg max \{\Delta  Q(u,C_s(u),C_s(v)),\Delta Q(v,C_s(v),C_s(u))$, \\
$Q(\{u,v\},\{C(u),C(v)\},C^{\ast})\}$\\
Move $w$ to a new community and update $C_S$\\
\For{$t\in N(w)$}{
Let $t$ decide its own community 
Update $C_s$
}
}
}
\Return $V_s,E_s,C_s$   
\end{function1}
\if{0}
\begin{function2}[!t]\tiny
\caption{\small$OneinSampleOneNew(u,v,e_t,V_s,E_s,\mathcal{H},C_s)$}
\label{onesampleonenew}
\If{$\mathcal{H}$ is not full}{
Insert $v$ to $\mathcal{H}$\\
}
\Else{
Choose a node $x$ with $\mathcal{P}(x) \in V_s$ from $\mathcal{H}$ preferentially based on $\mathcal{H}_c$\\
$V_s,E_s,C_s,\mathcal{H} = OneinSampleOneinBuffer(\mathcal{P}(x),x,V_s,E_s,\mathcal{H},C_s)$\\
}

$\mathcal{H}_c[v]=1$\\ 
$\mathcal{H}_p[v]=u$\\ 
\Return $V_s,E_s,C_s,\mathcal{H}$   
\end{function2}
\fi

\begin{function2}[!t]\tiny
\caption{\small $NodeinBuffer(u,\mathcal{H})$}
\label{nodebuffer}
$\mathcal{H}_c[x]=\mathcal{H}_c[x] + 1$\\
\Return $\mathcal{H}$
\end{function2}


\noindent\textbf{(ii) \underline{Both $u$ and $v$ are in buffer:}}
The only action (lines 16 - 18 of Algorithm 1) taken is that in the buffer $\mathcal{H}$,   $\mathcal{H}_c$ entries of $u$ and $v$ are incremented by 1 which is achieved through the function $NodeinBuffer()$ (executed twice with $u$ and $v$).
Example: (edge $\{k,n\}$ in Figure \ref{fig_algo}).

\noindent\textbf{(iii) \underline{$u$ is in sample and $v$ is in buffer:}} 
In this case (edge $\{d,i\}$ in Figure \ref{fig_algo}) also the only action (lines 19 - 20 of Algorithm 1) taken is that in the buffer $\mathcal{H}$, $\mathcal{H}_c$ entry of $v$ is incremented by 1 which is implemented through the function $NodeinBuffer()$.

\subsubsection{Entry of a new node:} In the three subsequent cases, at least one node is neither present in the buffer or the sample (new). This node triggers a rearrangement, whereby, another selected node is removed from the buffer to make space for the new node, and this selected node is inserted into the graph sample $G_s$  which further triggers a rearrangement of the sample in case it has already reached its size limit ($n$). The function $NodeisNew()$ is invoked to accomplish this task. The rearrangements that take place are described next.



\noindent\underline{Remove node from buffer:} This is triggered when $\mathcal{H}$ is full and in order to make room for the new node one of the existing nodes need to be removed from $\mathcal{H}$. To this aim 
we preferentially remove $x$ from $\mathcal{H}$ based on the counts in $\mathcal{H}_c$ with the 
additional constraint that $\mathcal{P}(x)$ is present in $G_s$. We add node $x$ and edges $\{\mathcal{P}(x)$,$x$\}, into $G_s$. This is achieved by executing the function $RemoveNodefromBuffer()$. 

\begin{function3}[!t]\tiny
 \caption{\small $RemoveNodefromBuffer(\mathcal{H},V_s,E_s,C_s)$}
 \label{removenodebuffer}
 Choose $x$ preferentially from $\mathcal{H}$ $\mathcal{P}(x) \in V_s$ \\
 Remove $x$ from $\mathcal{H}$ \\
 $InsertNodeinSample(x,\mathcal{P}(x),V_s,E_s,C_s)$ \\
 \Return  $\mathcal{H},V_s,E_s,C_s$ 
\end{function3}

\noindent\underline{Selection of node for removal from $G_s$:}
Insertion of a node into $V_s$  (obtained in the previous step), 
necessitates the removal of an existing node from the sample ($V_s$) to make space for the new entry. 
Nodes with the lowest degree in $G_s$ are candidates for deletion. Among these candidate nodes the one (say $x$) with  the lowest clustering coefficient is then removed from the $G_s$ to allow insertion of a new node (selected in the previous step). Subsequently, all the edges incident on $x$ are removed from $G_s$. The function $CheckResizeSample()$ implements this task. Finally, the selected node ($x$) is inserted into $V_s$ utilizing the function $InsertNodeinSample()$, whereby, an edge $(x,\mathcal{P}(x))$ is added to $V_s$ and $x$ is assigned the community of $\mathcal{P}(x)$. 
\vspace{-3mm}
\begin{function4}\tiny
\caption{\small$CheckResizeSample(V_s,C_s,n,m)$}
\label{resizesample}
\If{$V_s == n$}{
Remove $m$ nodes say, $u_1, u_2,\cdots, u_m$  (and all their adjacent edges) from $G_s$ having lowest degree and clustering coefficient\\
\For{$u\in \{u_1,u_2,\cdots,u_m\}$}{
$C_s \leftarrow CommunityAfterNodeRemoval(u,C_s)$
}
}
\Return $V_s,E_s,C_s$
\end{function4}
\vspace{-6mm}
\begin{function5}[!t]\tiny
 \caption{\small $InsertNodeinSample(x,\mathcal{P}(x),V_s,E_s,C_s)$}
 \label{insertnodesample}
 $V_s,E_s,C_s=CheckResizeSample(V_s,E_s,C_s,n,1)$ \\
 $V_s = V_s \cup x$ \\
 $E_s = E_s \cup \{x,\mathcal{P}(x)\}$ \\
 $C_s(x) = C_s(\mathcal{P}(x))$ \\
 Update $C_s$
 \Return $V_s,E_s,C_s$
\end{function5}

\noindent \underline{Adjust communities after removing a node:} 
Deletion of a node might keep the previous community structure unchanged, or break the community into smaller parts, or merge several communities together. The community structure $C_s$ is adjusted using  $Community -  AfterNodeRemoval()$ (Function \ref{communityadjust}) incrementally.  In the extreme,
removal of a node might render the community disconnected 
or broken into smaller parts which might further merge to the other existing communities~\cite{pone.0091431}. Here we utilize the clique percolation method~\cite{PalEtAl05} to handle this situation. In particular, when a vertex $v$ is removed from a community $C$, we place a 3-clique on one of its neighbors and let the clique percolate until no vertices in $C$ are discovered. Nodes discovered in each such clique percolation will form a community. We repeat this clique percolation from each of $v$'s neighbors until each member in $C$ is assigned to a community. For example, in Figure \ref{fig_percolation} when node $g$ is removed, we place a 3-clique on its neighbor $a$. Once the 3-clique starts percolating, it accumulates all nodes except $f$. Therefore, two new communities $\{a,b,c,d,e\}$ and $\{f\}$  emerge due to the deletion of $g$. In this way, we let the remaining nodes of $C$ choose their best communities to merge in. 
\begin{function6}\tiny
\caption{\small$CommunityAfterNodeRemoval(u,C_s)$}
\label{communityadjust}
Assume node $u$ and its adjacent edges are removed from $G_s$\\
$i=1$\\
\While{$N(u)\neq \phi$}{
$b_i$=Nodes found by a 3-clique percolation on $v\in N(u)$\\
\If{$b_i==\phi$}{
$b_i=\{v\}$
}
$C_s=C_s \cup b_i$\\
$N(u)=N(u)\setminus b_i$\\
$i=i+1$\\
}
Update $C_i$\\
\Return $C_s$   
\end{function6}

We now proceed to discuss the remaining cases. \\
\noindent\textbf{(iv) \underline{$u$ is in sample and $v$ is new:}}
In this case (handled by lines 21 - 22 in Algorithm 1) $v$ is inserted into the buffer $\mathcal{H}$ if $\mathcal{H}$ is not full. Otherwise its insertion triggers rearrangements of $\mathcal{H}$ and subsequently $V_s$. We use $NodeisNew()$ to accomplish this task. 

\noindent\textbf{(v) \underline{$u$ is in buffer and $v$ is new:}} In this case (edge $\{m,p\}$ in Figure \ref{fig_algo}), we increment the counter corresponding to $u$ and attempt to insert $v$ into $\mathcal{H}$ using the function $NodeisNew()$. 

\begin{function7}[!t]\tiny
 \caption{\small $NodeisNew(u,v,\mathcal{H},V_s,E_s,C_s)$}
 \label{nodenew}
 \If{$\mathcal{H}$ is full}{
$RemoveNodefromBuffer(\mathcal{H},V_s,E_s,C_s)$ \\
}
Insert u in $\mathcal{H}$ \\
$\mathcal{H}_c[u] = 1$\\
$\mathcal{H}_p[u] = v$ \\

\Return $\mathcal{H},V_s,E_s,C_s$
\end{function7}

\noindent\textbf{(vi) \underline{Both $u$ and $v$ are new:}}
In this case we attempt to insert both $u$ and $v$ to the buffer by executing the function $NodeisNew()$. 

Summarizing, the algorithm continuously increases the proportion of high fidelity nodes and 
improves the community structure by the following actions -- 
(a) delaying the insertion of a node to the sample allows for determining the importance of a node.\\
(b) removal of low clustering coefficient nodes from the sample ensures that only nodes with high clustering coefficient constitute the final $G_s$.  \\
(c) since all the actions are aimed at improving  modularity at every iteration, the final $G_s$ potentially will have 
 well-separated community structure.

\begin{figure}[!t]
\centering
\vspace{-4mm}
\includegraphics[width=\columnwidth]{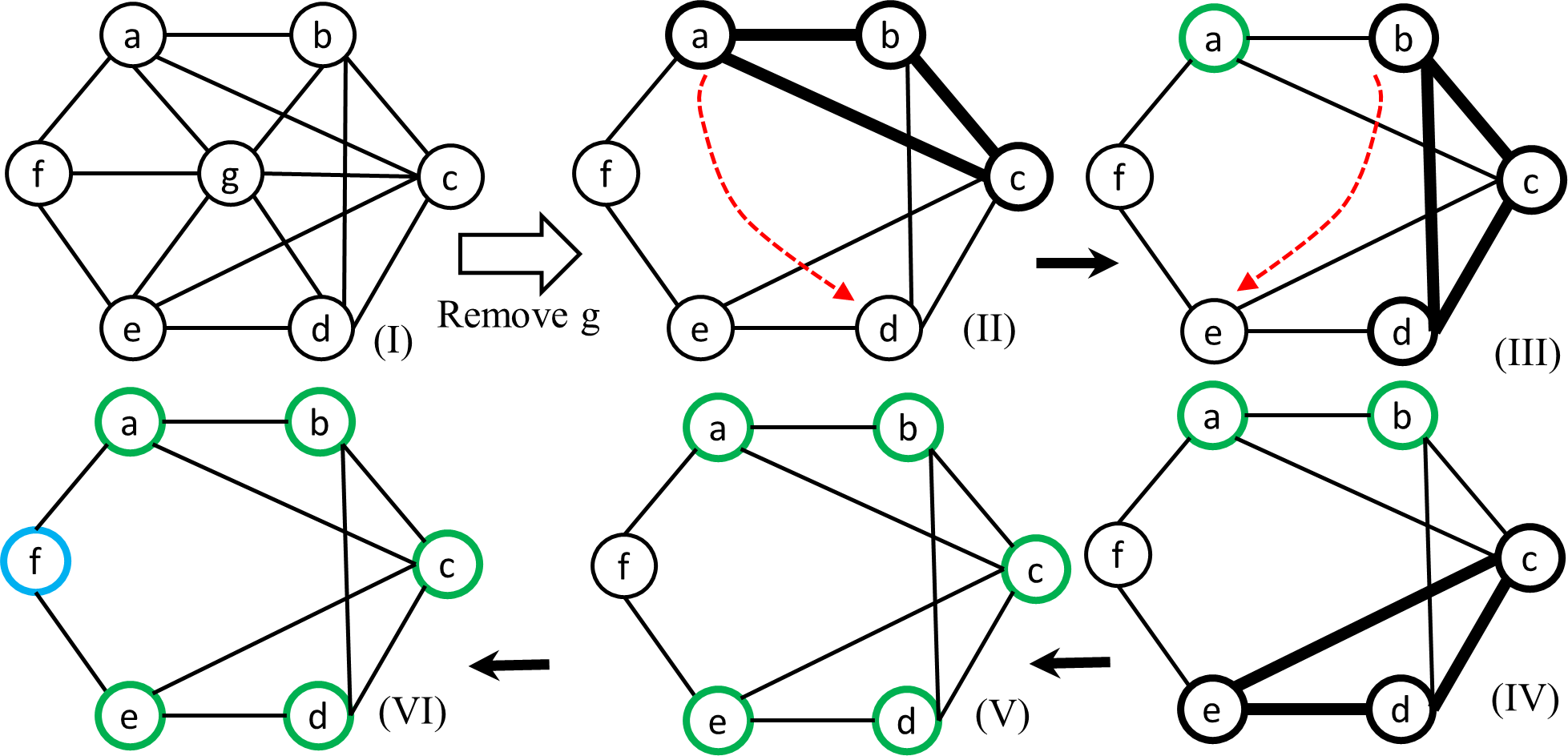}
\caption{Illustrative example of 3-clique percolation. Once node $g$ is removed, a 3-clique is placed on node $a$. The clique  percolates and accumulates all the nodes except node $f$ which forms a singleton community along with $\{a,b,c,d,e\}$.\vspace{3mm}}\label{fig_percolation} 
\end{figure}


\section{Experimental setup}
In this section, we outline the baseline sampling algorithms and the datasets used in our experiments.

\noindent{\bf Sampling algorithms:} 
We compare \compas~ with five existing sampling methods: (i) Streaming Node (SN)~\cite{ahmed2014network}, (ii) Streaming Edge (SE)~\cite{ahmed2014network}, (iii) Streaming BFS (SBFS)~\cite{ahmed2014network}, (iv) PIES~\cite{ahmed2014network}, and (v) Green Algorithm (GA)~\cite{tong2016novel}. The first four algorithms are exclusively designed for streaming graphs while the last one is designed for static graphs. Note that unlike ours, none of the existing methods explicitly produce a community structure as a by-product of the sampling,
and thus one needs to execute community detection algorithm separately on the sample to obtain the community structure. Therefore to evaluate the competing methods w.r.t how the underlying community structure in the sample corresponds to that of the original graph, for SN, SE, SBFS and PIES we run the Louvain algorithm~\cite{blondel2008fast}
\footnote{We also considered other algorithms (CNM~\cite{clauset2004finding}, GN~\cite{girvan2002community} and Infomap~\cite{rosvall2008maps}) and found the results to be similar.} 
on each individual sample and detect the communities. In case of GA, we consider the aggregated graph and run GA to obtain the sample, and further run Louvain algorithm on the sample to detect the community structure. 
Note that although the use of aggregated graph allows GA to leverage considerably more information about the graph structure, we use it as a strict baseline in this study.

\begin{table}[]
\centering
\caption{Datasets used for evaluation.}
\label{tab:data}
\begin{adjustbox}{max width=0.5\textwidth}
\begin{tabular}{l l l l l l}
\hline
Dataset  & Facebook & arxiv hep-th & Youtube   & Dblp      & LFR     \\ \hline
\# Nodes & 63,731   & 22,908       & 1,134,890 & 317,080   & 25,000  \\ 
\# Edges & 817,035  & 2,444,798    & 2,987,624 & 1,049,866 & 254,402 \\ \hline
\end{tabular}
\end{adjustbox}
\vspace{2mm}
\end{table}

\noindent{\bf Datasets:} 
We perform our experiments on the following five graphs (the first two are streaming and last three are static): \\ 
(i) {\bf Facebook}\footnote{konect.uni-koblenz.de/networks/facebook-wosn-links}: 
An undirected graph where nodes (63,731) are users, and edges (817,035) are friendship links that are time-stamped.\\ 
 (ii) {\bf arxiv hep-th}\footnote{konect.uni-koblenz.de/networks/ca-cit-HepTh}:  
Here nodes (22,908) are authors of arXiv's High Energy Physics papers and an edge exists between two authors if they have co-authored a paper; edges (2,444,798) are time-stamped by the publication date.\\ 
 (iii) {\bf Youtube}\footnote{snap.stanford.edu/data/com-Youtube.html}: Here 
nodes (1,134,890) represent Youtube users and edges (2,987,624) represent friendship. \\
(iv) {\bf dblp}\footnote{snap.stanford.edu/data/com-DBLP.html}:  
This dataset consists of authors indexed in DBLP. The graph is same as arxiv hep-th (317,080 nodes and 1,049,866 edges).\\
(v) {\bf LFR}~\cite{lancichinetti2008benchmark}: This is a synthetic graph with underlying community structure implanted into it.
We construct the graph with 25,000 nodes, 254,402 edges and 1,834 communities.\\
Since the last three graphs are static, we consider that each edge arrives in a pre-decided (random) order, i.e., each edge has a (discrete) time of arrival. The edge ordering, as we shall see, does not influence the inferences drawn from the results (Section~\ref{sec:effect}).
Moreover, since the first four graphs do not have any underlying ground-truth community structure, we run Louvain algorithm 
on the aggregated graph and obtain the disjoint community structure. This community structure is the best possible output that we can expect from our incremental modularity maximization method, and therefore serves as the ground-truth. The details of the datasets are summarized in Table~\ref{tab:data}.

\section{Evaluation}

\begin{table*}[!t]
\centering
\caption{\label{tab_all}Summary of the $D$-statistics (the lower, the better) values of the topological measures for all the datasets. For Youtube we present all the results, while for the rest we provide the average $D$-statistics and standard deviation (SD). Detailed results on other datasets can be found in \cite{si}. \compas~truns out to be the second best algorithm after GA (the most informed static graph sampling algorithm for which the sample is obtained from the aggregated graph and Louvain is run on the sample, thus serving as the strict baseline). Top two values for each average result is highlighted.}

\begin{adjustbox}{max width=\textwidth}
\begin{tabular}{l|c c c c c c c c c c c c c |c |c|c|c|c}
\hline
 \multirow{2}{*}{Algorithm}          & \multicolumn{14}{c|}{Youtube}                                                       & Facebook  & dblp & LFR & hep-th    \\ \cline{2-19}
 & ID & EI & AD & FOMD & TPR & EX & CR & CON & NC & AODF & MODF & FODF & MOD & Avg,SD & Avg,SD & Avg,SD  & Avg,SD & Avg,SD \\ \hline
\compas     & 0.063   & 0.051   & 0.078    & 0.057     & 0.227    & 0.082   & 0.054   & 0.091    & 0.260   & 0.073    & 0.201     &  0.121    & 0.052    & {\bf 0.10,0.07}   & {\bf 0.17,0.09}  &  {\bf 0.16,0.10}   & {\bf 0.18,0.06} &  {\bf 0.10,0.03 }   \\ 
SN         & 0.164   &  0.171  & 0.471   & 0.061     & 0.542    & 0.581   & 0.112   & 0.265    & 0.064   & 0.157     & 0.182     & 0.092     &  0.216   &  0.23,0.17  &       0.33,0.17  &    0.29,0.20      &  0.27,0.07 &  0.26,0.04    \\ 
SE         &  0.257  & 0.244   & 0.241   & 0.501     & 0.281    & 0.098   & 0.287   & 0.087    & 0.151   & 0.097     &  0.246    &  0.093    & 0.198    &  0.21,0.11  &       0.27,0.11  &   0.25,0.14       &   0.32,0.08  & 0.29,0.06   \\ 
SBFS       &  0.126  & 0.131   & 0.172   &  0.106    & 0.454    & 0.145   & 0.056   & 0.165    & 0.045   &   0.257   & 0.108     & 0.076     & 0.181    &  0.15,0.10  &        0.26,0.09  &  0.24,0.10        &  0.25,0.09 &  0.26,0.04    \\ 
PIES       &  0.234  & 0.241   &  0.252  & 0.190    & 0.409    & 0.042   & 0.051   & 0.049    & 0.061   &      0.157 & 0.042     & 0.053     & 0.121    &  0.14,0.10  &         0.29,0.06  &   0.24,0.07       &  0.26.0.05  &  0.21,0.05   \\ 
GA         &   0.156 & 0.055   & 0.065   & 0.053     & 0.267    & 0.066   & 0.076   & 0.053    & 0.085   &   0.150   & 0.075     &  0.069    & 0.102    &  {\bf 0.09,0.06} & {\bf 0.12,0.04}   &   {\bf 0.12,0.06}  &  {\bf 0.14,0.06} &  {\bf 0.08,0.04}    \\ \hline
\end{tabular}
\end{adjustbox}
\end{table*}

In this section, we list the standard metrics used to evaluate the goodness of the community structure, followed by a detailed comparison of the sampling algorithms.

\noindent{\bf Evaluation criteria:}
To measure how sampling algorithms capture the underlying community structure, we evaluate them in two ways.  First we measure the quality of the obtained community structure based on the {\bf topological measures} defined by \cite{yang2015defining}. In particular, we look into four classes of quality scores - (i) {\em based on internal connectivity}: internal density (ID), edge inside (EI), average degree (AD), fraction over mean degree (FOMD), triangle participation ratio (TPR); (ii) {\em based on external connectivity}: expansion (EX), cut ratio (CR); (iii) {\em combination of internal and external connectivity}: conductance (CON), normalized cut (NC), maximum out-degree fraction (MODF), average out-degree fraction (AODF), flake out-degree fraction (FODF); and (iv) {\em based on graph model}: modularity (MOD).
Note, for every individual community we obtain a score, and therefore a distribution of scores (i.e., distribution of ID, EI etc.) is obtained for all the communities of a graph. We measure how similar (in terms of Kolgomorov-Smirnov $D$-statistics\footnote{It is defined as $D = max_x\{|f(x) - f^{'}(x)|\}$ where $x$ is over the range of the random variable, and $f$ and $f^{'}$ are the two empirical cumulative distribution functions of the data.}) these distributions are with those of the ground-truth communities. {\em The lesser the value of D-statistics, the better the match between two distributions}.

\noindent{\bf Parameter estimation}:
As reported in Section~\ref{algorithm}, \compas~consists of two parameters: (i) $\alpha$ (initial fraction of nodes inserted), (ii) $n_d$ (length of the buffer).
We observe that $D$-statistics is initially high and reduces as we increase $\alpha$ (Figures \ref{param_est}(a)). For low $\alpha$, the community structure obtained initially by running a community-detection algorithm (line 12 in Algorithm 1) is coarse. For larger values of $\alpha$ even though initial community structure obtained is good, it is not allowed to evolve much. Similarly, in Figure~\ref{param_est}(b), given a small buffer size several nodes mostly arriving once would be added to the sample leading to formation of pendant vertices. As we increase the buffer size \compas~performs better till a certain point, after which the improvement is negligible. Since we are 
constrained by space, {\bf we fix $n_d$ at $0.0075 n$.} Similarly {\bf $\alpha$ is set to $0.4$.} 
We also set {\bf $n$ to $0.4|V|$} as default (see Section \ref{sec:effect} for different values of $n$). Further note that apart from Louvain we also consider other algorithms (CNM~\cite{clauset2004finding}, GN~\cite{girvan2002community} and Infomap~\cite{rosvall2008maps}) for obtaining the initial community structure. 
The average $D$-statistics values (calculated for LFR) across all the quality scores for Louvain, CNM, GN and Infomap are respectively \fbox{\textbf{0.182}, \textbf{0.191}, \textbf{0.216} and \textbf{0.197}}. 
Above results indicate that the quality of the initial communities are largely independent of the algorithm used. So we 
stick to the most popular one - Louvain for evaluation.

\begin{figure}[!h]
\vspace{-3mm}
\centering
\includegraphics[width=\columnwidth]{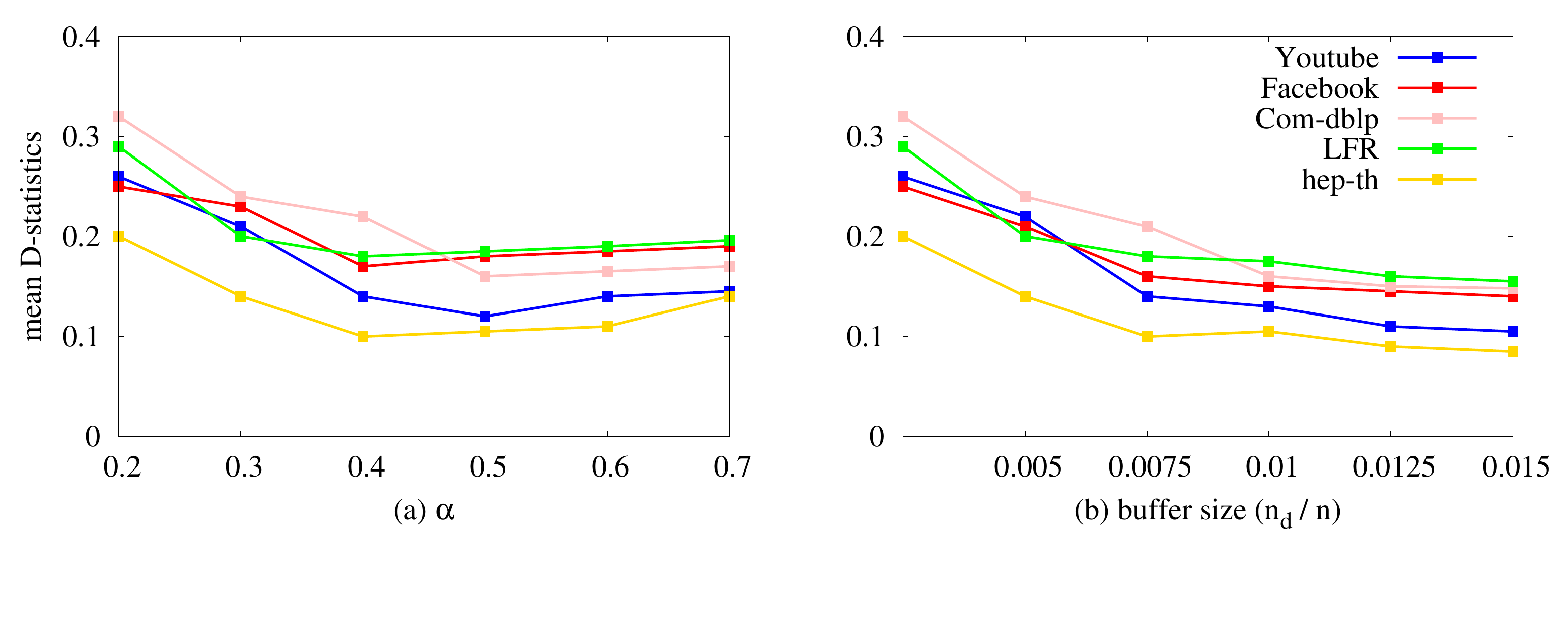}
\vspace{-9mm}
\caption{\label{param_est}Average $D$-statistics value across all the topological measures for various values of $\alpha$ and $n_d$.}
\end{figure}
Since the nodes are labeled, as a second level of evaluation, we use the {\bf community validation metrics} --  Purity~\cite{manning2008introduction}, Normalized Mutual Information (NMI)~\cite{danon2005comparing} and Adjusted Rand Index (ARI)~\cite{hubert1985comparing} to measure the similarity between the ground-truth and the obtained community structures. {\em The more the value of these metrics, the higher the similarity.}

\begin{table}[!t]
\centering
\caption{\label{g_metric_alg}NMI between the ground-truth and community structure obtained from individual sampling algorithms for all datasets.\vspace{2mm}}
\scalebox{0.8}{
\begin{tabular}{l cccccc}
\hline
\multirow{1}{*}{{\bf Dataset}} & \multicolumn{1}{c}{{\bf \compas}} & \multicolumn{1}{c}{{\bf SN}} & \multicolumn{1}{c}{{\bf SE}} & \multicolumn{1}{c}{{\bf SBFS}} & \multicolumn{1}{c}{{\bf PIES}} & \multicolumn{1}{c}{{\bf GA}} \\\hline
Facebook	  & 0.52 & 0.34 & 0.28&0.41&0.48&0.61\\
hep-th		  &0.51&0.32&0.21&0.36&0.39&0.68\\
Youtube		      &0.72&0.49&0.33&0.58&0.51&0.77 \\
dblp		  &0.65&0.28&0.21&0.57&0.39&0.69 \\
LFR		  &0.69&0.29&0.32&0.38&0.31&0.72\\\hline
Average & {\bf 0.61}  & 0.34 & 0.27 & 0.46  & 0.41 & {\bf 0.69}    \\
\hline
\vspace{4mm}
\end{tabular}}
\end{table}
\vspace{2mm}

\noindent{\bf Comparison of sampling algorithms}:
We start by measuring the similarity between the obtained and the ground-truth community structures using topological measures. In Table~\ref{tab_all} we 
summarize the $D$-statistics values of all the scoring functions for the Youtube dataset; for the other graphs we only present the average value (and standard deviation) 
across the $D$-statistics for different topological measures (detailed results on other datasets can be found in the \cite{si}). 
Since GA is specifically designed for static graphs, we simulate GA on the aggregated network consisting of every edge that has arrived, thereby allowing it 
more information compared to the other (streaming) algorithms which never have the whole graph under consideration. 
Clearly~\compas~outperforms all the streaming algorithms across different datasets and conceivably  
GA performs better than~\compas~as apart from utilizing the 
whole network structure, it further utilizes clustering coefficient and Pagerank of each node to obtain the sample.
Further we find \compas~is the second ranked algorithm after GA with an average (over all datasets) purity, NMI and ARI of \fbox{{\bf 0.74}, {\bf 0.61} and {\bf 0.53}} respectively (see Table~\ref{g_metric_alg}  for NMI, details in \cite{si}). Thus, \compas~matches the ground truth community both structurally and in content. 

{Among the rest of the sampling algorithms PIES performs best as it is biased towards the high degree nodes but at no point attempts to maximize modularity or clustering coefficient. The limited observability of graph structure using  a window in case SBFS, renders it ineffective in properly sampling high fidelity nodes.
For SN since nodes are picked uniformly at random the nodes with low degreee are shortlisted.  Similarly for SE, edges are picked uniformly at random and is again not inclined to pick nodes with any specific property. Hence SN and SE perform poorly in the task of preserving community structure.} 

\noindent{\bf Effect of edge ordering and sample size:}\label{sec:effect}
In this section, we show that most of our inferences are valid irrespective of any edge ordering. We randomly pick one pair of edges and swap their arrival time. 
We repeat it for $y$\% of edges (where $y$ varies between 5 and (as high as) 50) present in each aggregated graph. 
For each such ordering we obtain a representative sample (say $G_y$) and compare (average $D$-statistics)  with 
the ground-truth community. 
In figure~\ref{param_est_1}(a) we plot the $D$-statistics value averaged over all the scoring functions for the Youtube dataset. The plot clearly shows that the edge-ordering  
affects the final sample marginally (the pattern is same for other graphs).

Lastly, we present the effect of sample size ($n$) on the obtained community structure. We plot average $D$-statistics values across all the topological measures for all the algorithms on Youtube  (see others in \cite{si}) as a function of $n$ (Figure \ref{param_est_1}(b)). As expected, with the increase of $n$ we obtain better results. Interestingly, for~\compas~and GA, the pattern remains consistent compared to others. Moreover as we increase $n$ the divergence between their performance decreases.\\

\vspace{-3mm}
\begin{figure}[!h]
\centering
\includegraphics[scale = 0.28]{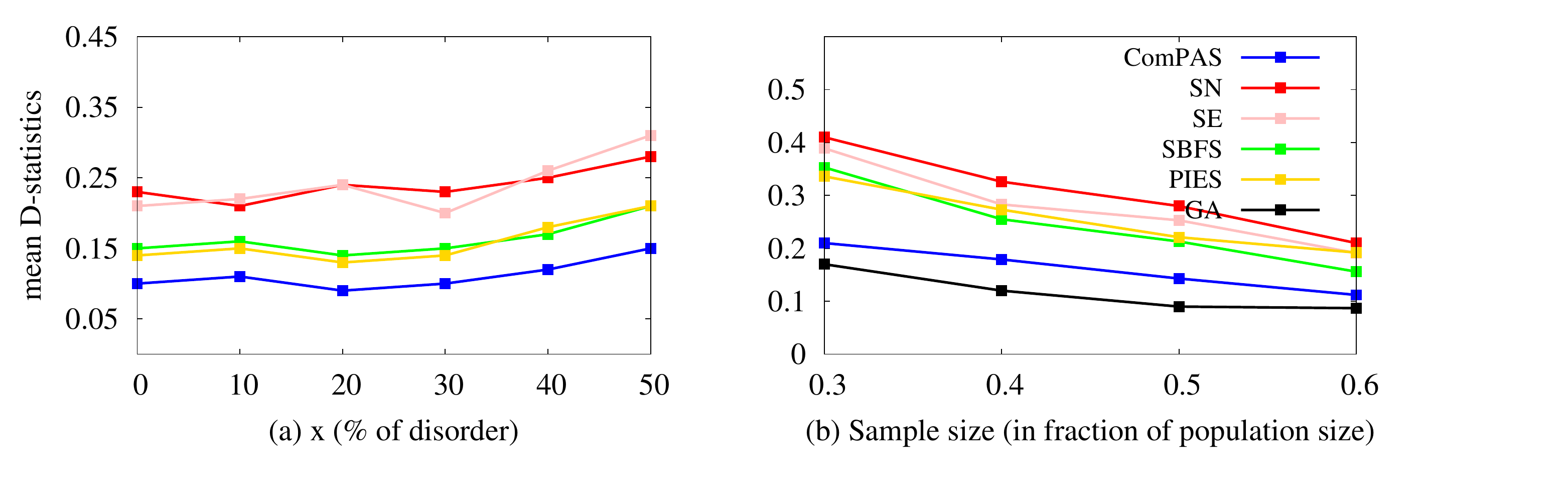}
\vspace{-10mm}
\caption{\label{param_est_1}Average $D$-statistics across all the topological measures for (a) different edge ordering and (b) sample size ($n$) of the Youtube graph.\vspace{-5mm}}
\end{figure}

\section{Complexity analysis}

\begin{table}[]
\centering
\caption{Machine specifications used for experiments.}
\label{tab:spec}
\begin{tabular}{l |l| l| l}
\hline
RAM   & CPU                                                                   & OS               & Cores \\ \hline
64 GB & \begin{tabular}[c]{@{}l@{}}Intel Xeon X5690\\ @ 3.47 GHZ\end{tabular} & Ubuntu 12.04 LTS & 24    \\ \hline
\end{tabular}
\end{table}

We perform two sets of experiments to determine the scalability of the algorithm - (i) dependence on stream size (total number of edges arriving in a single pass of the stream) 
and (ii) dependence on graph size ($N$). 
We stress that the complexity of the algorithm is (almost) linear with the size 
of the stream as at every step we perform certain local operations (depending on the case encountered) namely calculating modularity and clustering coefficient (calculated only
for low degree nodes during deletion).   
As a proof of concept, we consider an LFR graph with 25000 nodes and generate a sample of size 7500 with increasing stream sizes. 
In figure \ref{complexity}(a) we plot the time required for generating the sample. We note the machine specifications in Table \ref{tab:spec}.
We observe a linear behavior which corroborates our hypothesis.  
We further look into dependence on the size of the graph as well. In this regard we consider graphs of increasing sizes and measure the time required to obtain a sample of size 30\% of the population (refer to figure \ref{complexity}(b)). We again observe a linear behavior for the same machine specifications noted in Table~\ref{tab:spec}.
The above results hence indicate that \compas~ is scalable for large graphs as well.   

\begin{figure}
\centering
\includegraphics[scale=.25]{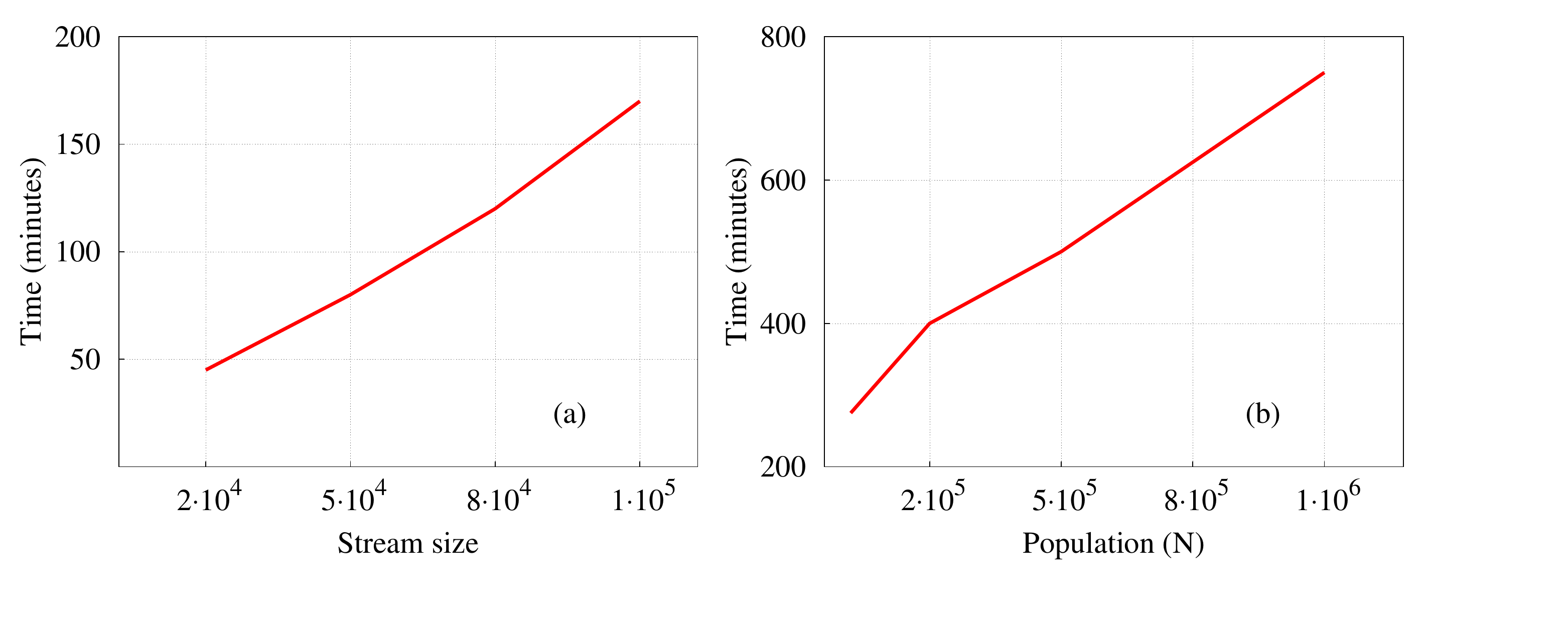}
\vspace{-6mm}
\caption{\label{complexity} Execution time of \compas~with increasing (a) stream size and (b) population size ($N$). A linear behavior is observed.\vspace{3mm}}
\end{figure}

\section{Insights}
In this section, we present certain micro-scale insights illustrating why \compas~outperforms the other algorithms in generating the community structure. \\
(i) \textbf{\compas~ admits high fidelity nodes and improves the modularity of the sample}: We observe how modularity, average clustering coefficient and average degree of the sample change over time as the edges arrive in a stream (refer to figure \ref{fig:discuss}(a)). All these factors increase over time. Here we report the results from the point the sample size ($n$) is reached for the first time up to the end of the stream.\\
\begin{figure}[!h]
\centering
\includegraphics[width=\columnwidth]{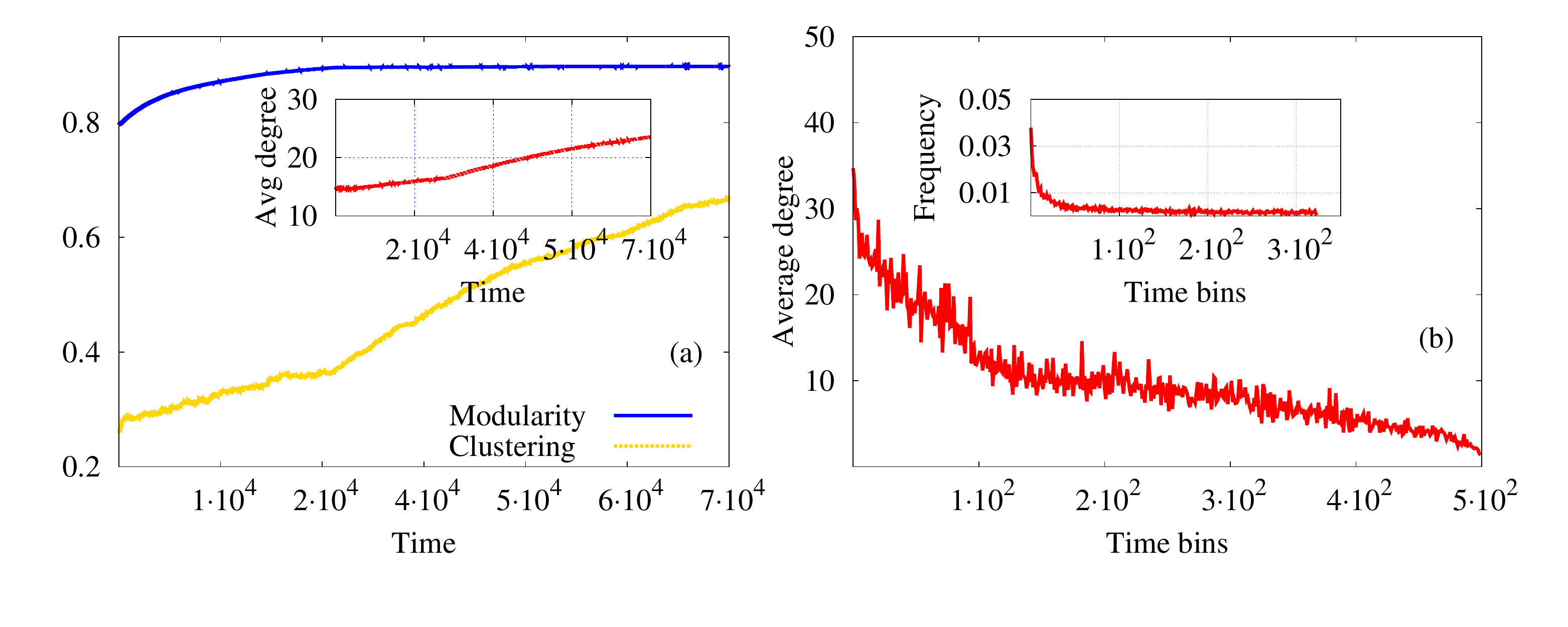}
\vspace{-8mm}
\caption{\label{fig:discuss}(a) Modularity and average clustering coefficient of the sample as it evolves over time, (inset) evolution of average degree of the sample over time.(b) Average degree of nodes in each bin (total time for streaming is divided into 500 equi-sized buckets), (inset) fraction of nodes in each bucket of the sample obtained using \compas.  The experiment is performed on Facebook dataset.\vspace{-2mm}}
\end{figure}
(ii) \textbf{\compas~ retains a large fraction of intra-community edges ensuring a better community structure}:  We observe that intra-community edges in the sample account for $\sim$ 80\% of all the edges while in the original network the corresponding value is $\sim$ 67\%. \\
(iii) \textbf{\compas~ produces a sample that has an edge density which corresponds highly to the original graph}: Note that \compas~is node-based, 
and $G_s$ consists of only those edges which
arrive {\em after} their corresponding nodes appear in $G_s$ - hence an efficient \compas~would 
insert the nodes  as early as possible. We compare the number of edges in $G_s$ against that in the subgraph ($\hat{G}_s$) induced by the sampled nodes in the original graph.
We observe that on average $G_s$ retains $\sim71$\% of the edges of $\hat{G}_s$. 
This indicates that the insertion time of nodes (in $G_s$) compared to their first 
appearance in the stream is early as $G_s$ is 
 able to retain most of the possible edges. \\
(iv) \textbf{\compas~samples high fidelity nodes uniformly over the time stretch}:
\compas~ samples more high fidelity nodes in time stretches where such nodes appear more frequently compared to the other stretches. To this purpose we split the stream into a set of buckets and a node is placed into a bucket based on the time it first arrived and 
calculate the average degree of each bucket (refer to figure \ref{fig:discuss}(b)). We observe that the average degree drops as we move from the first toward the subsequent buckets. We then consider the sample obtained from \compas~and calculate the fraction of sampled nodes in each bucket (figure \ref{fig:discuss}(b)(inset)). We observe a similar pattern indicating that \compas~ is not only able to sample the high degree nodes but the rate of sampling from each is roughly proportional to the average degree of each bucket.\\

\section{Applications of \compas~ in online learning}
In online learning, sometimes memory is limited and it is required to train the model on limited number of instances. One of the important problems in learning is
to judiciously choose the training sample set - a random sampling of edges do not produce a good representative set \cite{werner2012impact}.

We hypothesize that more diverse the chosen set, better would be the performance. \compas~ is useful in such cases since it tries to sample from several communities, hence improving the diversity of the training set. To this end, we consider Wiki-Rfa\footnote{https://snap.stanford.edu/data/wiki-RfA.html} \cite{west2014exploiting}, a streaming signed graph in which nodes represent Wikipedia members and edges (with time-stamp) represent votes. Each vote is typically accompanied by a short comment. The task is to predict the vote (+1, -1) of an incoming edge based on the textual features -- (i) word count, (ii) sentiment value, and (iii) LIWC features of the statement corresponding to the edge. 
Moreover, we can use certain extra features like whether the edge is an intra or inter community edge, the average degree and the clustering coefficient of the nodes connected with an edge etc. to train the model.
We allow training instances to be included till a certain time period $t$ (first 75\% of the edges are allowed to enter) and run the sampling algorithms in parallel. However not all instances can be considered for training due to the memory constraint. 
We assume $n$, the sample size as the allowed training size and obtain sampled training set from individual sampling algorithms. 
The size of the network is 4000 and that of the sample size is 1200 which is 30\% of the population.
We train SVM with linear kernel (see \cite{si} for other classifiers) on each sampled training set, and predict the labels (votes) of those instances coming after $t$. 
Table \ref{comp_wiki} shows that GA and \compas~perform the best in terms of AUC and F-Score. 
 This once again emphasizes  that \compas~selects most representative training instances for (restricted) online learning.
 \begin{table}[!t]
\centering
\caption{\label{comp_wiki} Performance of SVM using the training set obtained from sampling methods.}
\scalebox{0.8}{
\begin{tabular}{c|c|c|c|c|c|c}
\cline{1-7}
 & \compas & SN & SE & SBFS & PIES & GA \\\hline
AUC & \textbf{0.48} & 0.31 & 0.25 & 0.28 & 0.36 & \textbf{0.53} \\
F-Score & \textbf{0.61} & 0.35 & 0.28 & 0.31 & 0.43 & \textbf{0.64} \\
\hline
\end{tabular}}
\vspace{4mm}
\end{table}

\section{Discussion}
\if{0}
 \begin{table}[!t]
\centering
\caption{\label{comp_wiki} Performance of SVM using the training set obtained from sampling methods.\vspace{2mm}}
\scalebox{0.8}{
\begin{tabular}{c|c|c|c|c|c|c}
\cline{1-7}
 & \compas & SN & SE & SBFS & PIES & GA \\\hline
AUC & \textbf{0.46} & 0.31 & 0.25 & 0.28 & 0.36 & \textbf{0.55} \\
F-Score & \textbf{0.59} & 0.35 & 0.28 & 0.31 & 0.43 & \textbf{0.67} \\
\hline
\end{tabular}}
\vspace{2mm}
\end{table}

{\bf Complexity:}  Albeit our algorithm requires calculation of modularity and clustering coefficient at different time steps (depending on the case encountered), it is done only locally over the nodes which are affected by the rearrangements. Moreover clustering coefficient is calculated only over the nodes with low degree since it is used to decide which node to delete from the sample and the candidate nodes are already the ones with minimum degree. Further the degree for each node in the sample is maintained in the dictionary and updated as the stream progresses. Hence it is calculated in constant time. So the whole algorithm executes in linear time over the size of the stream. \noteng{Do we have time calculation from the machine}

\noteng{Predicting link values of a graph; Describe the graph; training set - created by choosing the edges which
become part of the sampled graph; using several features and then predicting - but it is not compared with standard algorithm.  }

\noindent{\bf Application to Online learning:} \noteng{Why this name} 
An immediate application of \compas~would to predict labels of edges in a  network \noteng{please specify} wikipedia administrator election~\cite{west2014exploiting} process 
where a member at certain time stamp casts a supporting or a opposing opinion about a user accompanied by a comment. This can be framed as a streaming graph  with each member representing a node and each vote an edge. The vote pattern (+ve/-ve) is predicted The task is to observe the stream (and thereby train) up to a certain time (first 75\% of the edges are allowed to enter) and subsequently predict the opinion of the future votes with an additional constraint that not all edges can be considered for training (due to limited memory) and the system has to resort to sampling a subset of incoming edges which is performed utilizing the streaming graph sampling algorithms. Apart from leveraging textual features namely word count, sentiment value, and LIWC features, we further consider clustering coefficient and degree of both the users involved as well as the edge type (intra or inter community). In case the node is new the degree is set as 1, clustering coefficient set to 0 and edge is considered as an inter community one. The size of the network is 4000 and that of the sample size (also the training size) is 1200 which is 30\% of the population.
We train SVM with linear kernel as our prediction model and observe that GA and \compas~perform the best in terms of AUC and F-Score (refer to table \ref{comp_wiki}). 
\fi

To conclude, we in this paper proposed \compas, a novel sampling algorithm for streaming graphs which is able to retain the community structure of the original graph. Through rigorous experimentation on real-world and synthetic graphs we showed that \compas~ performs better than four  state-of-the-art graph sampling algorithms. 
We also stress that the complexity of the algorithm is (almost) linear with the size of the graph as at every step we perform certain local operations (depending on the case encountered) namely calculating modularity and clustering coefficient (calculated only for low degree nodes during deletion). 

One of the important problems in learning is
to judiciously choose the training sample set and in this context, we demonstrated that \compas~can be used to shortlist the training sample. 
We would like to point out that, although encouraging, these are initial results. A thorough analysis needs to be done on each individual use-case before strong (and universal) claims can be advocated - this would exactly be our immediate future pursuit. 



\section{Appendix}
\label{appendix}
\subsection{Proof of propositions}

PROPOSITION 1.
{\em
Addition of an edge to a community $c\in C$, increases its modularity if $D_c\leq M-1$ $($where $M=|E|)$.}
\vspace{-2mm}
\begin{proof}
Recall the formulation of modularity as:
\vspace{-2mm}
\begin{equation}\label{modularity}\small
Q(G(V,E),C)=\sum_{c\in C} (\frac{m_c}{M} - \frac{D_c^2}{4M^2})
\vspace{-2mm}
\end{equation}
where $C$ is the community structure of $G$, $m_c$ is the total number of edges inside $c$, $D_c$ is the sum of degree of all the nodes inside a community $c\in C$, and $M=|E|$ is the total number of edges $G$.

From Equation \ref{modularity}, we see the contribution of individual community $c\in C$ in modularity as: $Q_c=\frac{m_c}{M} - \frac{D_c^2}{4M^2}$. 
where $m_c$ is the number of edges inside $c$, $M$ is the total number of edges in the graph, and $D_c$ is the sum of degrees of all the nodes in $c$. 

Addition of a new edge within $c$, the $c$'s contribution of modularity becomes:
\[\small
Q'_c=\frac{m_c+1}{M+1} - \frac{(D_c+2)^2}{4(M+1)^2}
\]

So the increase in modularity is $\Delta Q_c=Q'_c-Q_c$,
\[\small
\begin{split}
\Delta Q_c=&\frac{4M^2-4m_cM^2-4D_cM^2-4m_cM+2D_c^2M+D_c^2}{4(M+1)^2M^2}\\
&\geq \frac{4M^2-6D_cM^2-2D_cM+2D_c^2M+D_c^2}{4(M+1)^2M^2}\\
&\geq\frac{(2M^2-2D_cM-D_c)(2M-D_c)}{4(M+1)^2M^2}\\
&\geq 0
\end{split}
\]
The equality holds if $D_c\leq M-1$. This thus implies $(2M^2-2D_cM-D_c)\geq 0$. This proves the proposition.
\end{proof}

PROPOSITION 2.
{\em Addition of any intra-community edge into a community $c\in C$ would not split into smaller communities.}

\begin{proof}
We will prove this proposition by contradiction.
Assume that once a new intra-community edge is added into $c$, it gets split into $k$ small modules, namely $X_1$, $X_2$, $\cdot$,$X_k$. Let $D_{X_i}$ and $e_{ij}$ be the total degree of nodes inside $X_i$ and number of edges connecting $X_i$ and $X_j$ respectively.

Recall that the contribution of $X_i$ in the modularity value is $Q_{X_i}=\frac{m_{X_i}}{M} - \frac{D_{X_i}^2}{4M^2}$.  Before adding the edge, we have $Q_c \geq \sum_{i=1}^k Q_{X_i}$ (where $Q_c$ is the total modularity of community $c$), because otherwise all $X_i$s can be split earlier, which is not in this case. This implies that:
\[
\frac{m_c}{M}- \frac{D_c^2}{4M^2} > \sum_{i=1}^k (\frac{m_{X_i}}{M} - \frac{D_{X_i}^2}{4M^2})
\]
Since $X_1,X_2,\cdot,X_k$ are all disjoint modules of $c$, $D_c=\sum_{i=1}^k D_{X_i}$ and $m_c=\sum_{i=1}^k m_{X_i} + \sum_{i<j} e_{ij}$. This further implies that:
\[
\frac{m_c}{M}-\sum_{i=1}^k \frac{m_{X_i}}{M} > \frac{D_c^2}{4M^2}- \sum_{i=1}^k \frac{D_{X_i}^2}{4M^2}
\]
or,
\[
\sum_{i<j} e_{ij} > \frac{\sum_{i<j} D_{X_i}D_{X_j}}{2M}
\]
 
 Without loss of generality, let us assume that the new edge is added inside $X_1$. 
Since we assume that after adding the new edge into $c$, it gets split into $k$ small modules, the modularity value should increase because of the split. Therefore, the new modularity $Q'_c<\sum_{i=1}^k Q_{X_i}$. This implies that
\[\small
\begin{split}\small
& Q'_c < \sum_{i=1}^k Q_{X_i}\\
& \Leftrightarrow \frac{\sum_{i=1}^k m_{X_i} + \sum_{i<j} e_{ij} + 1}{M+1} - \frac{(\sum_{i=1}^k D_{X_i +2})^2}{4(M+1)^2}\\
& <  \frac{m_{X_1}+1}{M+1} - \frac{(D_{X_1}+2)^2}{4(M+1)^2}  + \sum_{i=2}^k  (\frac{m_{X_i}}{M+1} - \frac{D_{X_i}^2}{4(M+1)^2} )\\
&\Leftrightarrow \frac{\sum_{i=1}^k m_{X_i} + \sum_{i<j} e_{ij} + 1}{M+1} - \frac{(\sum_{i=1}^k D_{X_i +2})^2}{4(M+1)^2}\\
& < \frac{\sum_{i=1}^k m_{X_i}+1}{M+1} - \frac{(D_{X_1}+2)^2}{4(M+1)^2} - \sum_{i=2}^k \frac{D_{X_i}^2}{4(M+1)^2}\\
&\Leftrightarrow \sum_{i<i} e_{ij} < \frac{\sum_{i=1}^k D_{X_i} - 2 D_{X_1} + \sum_{i<j} D_{X_i}D_{X_j}}{2(M+1)}
\end{split}
\]
Since $\sum_{i=1}^k D_{X_i} - 2D_{X_1} < 2M$, this implies that 
\[\small
\begin{split}
\frac{\sum_{i<j}D_{X_i}D_{X_j}}{2M}  < \sum_{i<j}e_{ij} 
&< \frac{\sum_{i=1}^k D_{X_i} - 2D_{X_1} + \sum_{i\neq j}{D_{X_i}D_{X_j}}}{2(M+1)}\\
& <\frac{\sum_{i<j}D_{X_i}D_{X_j}}{2M}+1
\end{split}
\]
Therefore, the proposition holds.
\end{proof}




\end{document}